\documentclass{ecai}  %

\usepackage{graphicx}
\usepackage{latexsym}

\paperid{1770}        %
\usepackage{newfloat}
\usepackage{listings}

\usepackage[utf8]{inputenc} %
\usepackage[T1]{fontenc}    %
\usepackage{url}            %
\usepackage{booktabs}       %
\usepackage{amsfonts}       %
\usepackage{nicefrac}       %
\usepackage{microtype}      %
\usepackage[dvipsnames]{xcolor}         %

\usepackage{todonotes} %
\usepackage[utf8]{inputenc} %
\usepackage{url}            %
\usepackage{booktabs}       %
\usepackage{amsfonts}       %
\usepackage{nicefrac}       %
\usepackage{microtype}      %
\usepackage[dvipsnames]{xcolor}         %
\usepackage{xspace}
\usepackage{parcolumns}

\usepackage{adjustbox}
\usepackage{graphicx}
\usepackage{lipsum}

\usepackage{wrapfig}

\usepackage[utf8]{inputenc}
\usepackage{pgfplots}

\usepackage{microtype}
\usepackage{graphicx}
\usepackage{subfig}
\usepackage{booktabs} %

\usepackage{hyperref}

\usepackage{algorithm}
\usepackage[noend]{algpseudocode}

\PassOptionsToPackage{dvipsnames}{xcolor}

\usepackage{amsmath}
\usepackage{amssymb}
\usepackage{mathtools}
\usepackage{amsthm}
\usepackage{bm}

\usepackage{pgfplots}
\usepackage{tikz}
\pgfplotsset{
    compat=newest,
    error bars/y explicit,
    error bars/error bar style={solid},
}

\pgfplotscreateplotcyclelist{MyCyclelist}{%
blue,every mark/.append style={fill=blue!80!black},mark=*\\%
red,every mark/.append style={fill=red!80!black},mark=square*\\%
OliveGreen,every mark/.append style={fill=red!80!black},mark=otimes*\\%
black,mark=star\\%
orange,every mark/.append style={fill=orange!80!black},mark=diamond*\\%
BlueGreen,densely dashed,every mark/.append style={solid,fill=BlueGreen!80!black},mark=*\\%
brown!60!black,densely dashed,every mark/.append style={
solid,fill=brown!80!black},mark=square*\\%
black,densely dashed,every mark/.append style={solid,fill=gray},mark=otimes*\\%
blue,densely dashed,mark=star,every mark/.append style=solid\\%
red,densely dashed,every mark/.append style={solid,fill=red!80!black},mark=diamond*\\%
}

\usepackage{amssymb}
\usepackage{booktabs}
\usepackage{adjustbox}
\usepackage{wrapfig}

\newtheorem{theorem}{Theorem}[section]

\begin{document}

\begin{frontmatter}

\title{An Easy Rejection Sampling Baseline via Gradient Refined Proposals}

\author[A,B,C,D]{\fnms{Edward}~\snm{Raff}}
\author[A]{Mark McLean}
\author[A]{James Holt}

\address[A]{Laboratory for Physical Sciences}
\address[B]{Booz Allen Hamilton}
\address[C]{Syracuse University}
\address[D]{University of Maryland, Baltimore County}

\begin{abstract}
Rejection sampling is a common tool for low dimensional problems ($d \leq 2$), often touted as an ``easy'' way to obtain valid samples from a distribution $f(\cdot)$ of interest. In practice it is non-trivial to apply, often requiring considerable mathematical effort to devise a good proposal distribution $g(\cdot)$ and select a supremum $C$. More advanced samplers require additional mathematical derivations, limitations on $f(\cdot)$, or even cross-validation, making them difficult to apply. We devise a new approximate baseline approach to rejection sampling that works with less information, requiring only a differentiable $f(\cdot)$ be specified, making it easier to use. We propose a new approach to rejection sampling by refining a parameterized proposal distribution with a loss derived from the acceptance threshold. In this manner we obtain comparable or better acceptance rates on current benchmarks by up to $7.3\times$, while requiring no extra assumptions or any derivations to use: only a differentiable $f(\cdot)$ is required. While approximate, the results are correct with high probability, and in all tests pass a distributional check. This makes our approach easy to use, reproduce, and efficacious. 
\end{abstract}

\end{frontmatter}

\section{Introduction}

Given a target distribution $f(\cdot)$ that we wish to draw samples from, rejection sampling provides one of the most popular strategies. The approach is often advertised as ``easy,'' given the simplicity of its procedure. Given a proposal distribution $g(\cdot)$ that we do know how to sample from, and a value $C$ such that $C \cdot g(x) \geq f(x) \forall x$, rejection sampling can be described succinctly in three steps:

\begin{enumerate}
\item  Generate $x \sim g(\cdot)$ and $u \sim \mathcal{U}(0, 1)$, where $\mathcal{U}$ is the continuous uniform distribution. 
\item Accept $x$ as a valid sample of $f(\cdot)$ if $u \leq \frac{f(x)}{C g(x)}$.
\item If not accepted, go back to step 1.  
\end{enumerate}

This also makes $1/C$ directly interpretable as the acceptance rate of the sampling procedure, and it guarantees independent and identically distributed (I.I.D.) samples. For this reason rejection sampling is often used in modeling applications where a physical processes is known, but produces bespoke distributions of few variables that need to be sampled ~\cite{kermiche_total_2022,ozols_quantum_2012,jewson_adjusting_2020,mitchell-wallace_natural_2017,nguyen_acceptance-rejection_2014,paskov_faster_1998,Jul2021}. Though algorithmically simple, this belays the non-trivial amount of work required by a user who wishes to draw samples. First one must perform considerable work to devise a distribution $g(\cdot)$ to sample from, find a method to determine $C$ or find a small upper bound of $C$. 

Because of this difficulty, there has been considerable work over time attempting to perform rejection sampling with greater efficacy. These newer sampling strategies often suffer from additional constraints on the target distribution $f(\cdot)$, or additional work on the user to do additional derivations. In both cases the application of a rejection sampler is non-trivial for a user, and slows progress toward prototyping, building simulations, and various down-stream tasks like HMC and Gibbs samplers that may desire to leverage sampling for an intermediate step. 

\begin{figure}[!t]
    \centering
    \adjustbox{max width=0.99\columnwidth}{%
    \begin{tikzpicture}

\definecolor{color0}{rgb}{0.12156862745098,0.466666666666667,0.705882352941177}
\definecolor{color1}{rgb}{1,0.498039215686275,0.0549019607843137}
\definecolor{color2}{rgb}{0.172549019607843,0.627450980392157,0.172549019607843}

\begin{axis}[
legend cell align={left},
legend style={fill opacity=0.8, draw opacity=1, text opacity=1, draw=white!80!black},
tick align=outside,
tick pos=left,
x grid style={white!69.0196078431373!black},
xmin=-5.75533199580941, 
xmax=5.51215866646711,
xtick style={color=black},
 legend style={font=\tiny},
y grid style={white!69.0196078431373!black},
ymin=-0.00919013883502155, ymax=0.211264547108374,
ytick style={color=black},
xlabel=Input $x$,
ylabel=Probability $\mathbb{P}(x)$,
yticklabel style={
/pgf/number format/fixed,
/pgf/number format/precision=2
},
]
\addplot [draw=white, fill=color0, forget plot, mark=*, only marks]
table{%
x  y
-2.45408063787313 0.147304544494416
-4.89255956535183 0.0266205286059788
-4.99358897164254 0.0218753976057837
-4.06725898560156 0.090339067875623
-3.82671617251779 0.113517638466833
-3.38654514308717 0.148734482418585
-5.24317332934229 0.0128922881731769
-3.49861669313258 0.14136189107296
-1.28682451616164 0.123940171285403
-0.0826488809022778 0.201054852945502
0.0853628573840893 0.200112947532958
-0.535963859326452 0.18044972528718
1.15419405383195 0.102500353110527
-0.52335572871809 0.181372353460332
-0.572710189621526 0.177686588189846
-0.512970366965474 0.182121828012391
-0.757179225204211 0.162658583433792
0.519640840808635 0.174604611021675
3.65836576443745 0.0634252190990639
3.04092561241004 0.0146345264205772
};
\addplot [semithick, color0, dotted]
table {%
-5 0.0215971299650326
-4.9 0.0262475453902436
-4.8 0.0315820439699032
-4.7 0.0376228158633442
-4.6 0.044373404297815
-4.5 0.0518150301369101
-4.4 0.0599034574899429
-4.3 0.068566704417494
-4.2 0.0777038935271634
-4.1 0.0871855016416088
-4 0.0968552049205398
-3.9 0.106533427695198
-3.8 0.116022594567489
-3.7 0.12511396348198
-3.6 0.133595792121777
-3.5 0.141262472053243
-3.4 0.147924165705566
-3.3 0.153416410653737
-3.2 0.157609121690915
-3.1 0.160414428518905
-3 0.161792836366542
-2.9 0.161757285200693
-2.8 0.160374803381672
-2.7 0.157765593591421
-2.6 0.154099540738172
-2.5 0.149590280952504
-2.4 0.144487106304141
-2.3 0.139065092217505
-2.2 0.133613917527709
-2.1 0.128425897949716
-2 0.123783773064251
-1.9 0.119948778200359
-1.8 0.117149501143732
-1.7 0.115571975507566
-1.6 0.115351403594026
-1.5 0.116565836099303
-1.4 0.119232066689655
-1.3 0.123303926974658
-1.2 0.128673090811964
-1.1 0.135172414426146
-1 0.142581748864847
-0.899999999999999 0.150636063341548
-0.8 0.159035613519234
-0.7 0.167457781780021
-0.6 0.175570113563837
-0.5 0.183043983579577
-0.399999999999999 0.189568257845136
-0.3 0.194862281656031
-0.199999999999999 0.19868752762092
-0.0999999999999996 0.200857286706416
0 0.201243879565493
0.100000000000001 0.199783001360991
0.2 0.196474982268337
0.300000000000001 0.191382935305986
0.4 0.184627957819492
0.5 0.176381736461995
0.600000000000001 0.166857062225207
0.7 0.156296878819335
0.800000000000001 0.144962555252821
0.9 0.133122087496691
1 0.121038895565054
1.1 0.108961797133785
1.2 0.0971166170839501
1.3 0.0856997472498373
1.4 0.0748738170270953
1.5 0.0647654886732916
1.6 0.0554652659403721
1.7 0.0470291147889281
1.8 0.0394816521643415
1.9 0.0328206715851879
2 0.0270228439895529
2.1 0.0220505462727784
2.2 0.0178598912505523
2.3 0.0144100895584638
2.4 0.0116741570821882
2.5 0.00965056300821204
2.6 0.0083745996639208
2.7 0.0079270753097688
2.8 0.00843663975973713
2.9 0.0100711891854397
3 0.0130141199389598
3.1 0.0174234425154325
3.2 0.0233762117541372
3.3 0.0308067779807332
3.4 0.0394533207844047
3.5 0.0488304863581338
3.6 0.0582442615728206
3.7 0.0668573107421607
3.8 0.0738000015380843
3.9 0.0783078665381089
4 0.0798553711968229
4.1 0.0782531696254852
4.2 0.0736834995954261
4.3 0.0666641881771586
4.4 0.0579507817881467
4.5 0.0484021367744795
4.6 0.038842281422721
4.7 0.0299486780397596
4.8 0.0221861475854464
4.9 0.0157912511405562
5 0.010798936662397
};
\addlegendentry{Truth}
\addplot [semithick, color1, dashed]
table {%
-5 0.0337769988002927
-4.9 0.0391890856301241
-4.8 0.0450617060761197
-4.7 0.0513536585627463
-4.6 0.0580075188277943
-4.5 0.0649501189159881
-4.4 0.0720939049502544
-4.3 0.0793391969915272
-4.2 0.0865773125904092
-4.1 0.0936944499084158
-4 0.100576162009936
-3.9 0.107112196990843
-3.8 0.113201434736859
-3.7 0.118756625212261
-3.6 0.123708628792308
-3.5 0.128009877932988
-3.4 0.131636820935071
-3.3 0.134591170024644
-3.2 0.136899852759657
-3.1 0.138613651637296
-3 0.139804604531082
-2.9 0.14056232083128
-2.8 0.140989438049592
-2.7 0.141196495576372
-2.6 0.141296532431264
-2.5 0.141399722533074
-2.4 0.141608344746345
-2.3 0.142012348305034
-2.2 0.142685721407698
-2.1 0.143683807208062
-2 0.145041643000234
-1.9 0.146773330932565
-1.8 0.148872387236503
-1.7 0.151312965797068
-1.6 0.154051813631618
-1.5 0.157030791740462
-1.4 0.160179784718457
-1.3 0.163419825163009
-1.2 0.166666272092131
-1.1 0.169831903573353
-1 0.172829809649127
-0.899999999999999 0.175575999635747
-0.8 0.177991665540352
-0.7 0.18000506877486
-0.6 0.181553039232292
-0.5 0.18258209340963
-0.399999999999999 0.183049191422907
-0.3 0.182922161702327
-0.199999999999999 0.182179827417614
-0.0999999999999996 0.180811870980369
0 0.178818473073886
0.100000000000001 0.176209761310197
0.2 0.173005101447482
0.300000000000001 0.169232261618033
0.4 0.164926477565794
0.5 0.160129444676163
0.600000000000001 0.154888260678486
0.7 0.14925434129798
0.800000000000001 0.143282329750063
0.9 0.137029019695024
1 0.130552309986424
1.1 0.123910208147548
1.2 0.117159897917605
1.3 0.110356884377843
1.4 0.103554228086482
1.5 0.0968018773426022
1.6 0.090146105211665
1.7 0.0836290553485502
1.8 0.077288398029279
1.9 0.0711570952362998
2 0.0652632712186865
2.1 0.0596301827445153
2.2 0.0542762813427077
2.3 0.0492153582455168
2.4 0.0444567615237415
2.5 0.0400056740708529
2.6 0.0358634406392884
2.7 0.0320279320470259
2.8 0.028493934927108
2.9 0.0252535559485199
3 0.0222966302474784
3.1 0.0196111248225479
3.2 0.017183528811539
3.3 0.0149992238295258
3.4 0.0130428288545351
3.5 0.0112985154536449
3.6 0.00975029040605429
3.7 0.00838224396632478
3.8 0.00717876309271695
3.9 0.00612470992188137
4 0.00520556658877074
4.1 0.00440754816285746
4.2 0.00371768599793597
4.3 0.00312388417752937
4.4 0.00261495199004527
4.5 0.00218061549948663
4.6 0.00181151130320267
4.7 0.00149916550377718
4.8 0.0012359607841986
4.9 0.00101509428031263
5 0.000830528707860075
};
\addlegendentry{MSE: $C=22.2$}
\addplot [semithick, color2]
table {%
-5 0.0323958745727326
-4.9 0.0391420064293868
-4.8 0.0467601290230525
-4.7 0.0551605669913342
-4.6 0.0641891745779988
-4.5 0.0736276290264261
-4.4 0.0832006082606963
-4.3 0.0925902014222532
-4.2 0.10145689854078
-4.1 0.109465451244952
-4 0.116312977906749
-3.9 0.121756089981962
-3.8 0.125633683912196
-3.7 0.12788243746355
-3.6 0.128542933442402
-3.5 0.127755571818367
-3.4 0.125746814635934
-3.3 0.122807597450407
-3.2 0.119266716875604
-3.1 0.115462511805566
-3 0.11171613621158
-2.9 0.108309216004522
-2.8 0.105467818929144
-2.7 0.103353624845478
-2.6 0.102062155054464
-2.5 0.10162706659072
-2.4 0.10202894796438
-2.3 0.103206806448911
-2.2 0.105070490118809
-2.1 0.107512570112466
-2 0.110418625796296
-1.9 0.113675332310429
-1.8 0.117176167773461
-1.7 0.120824884448791
-1.6 0.124537101718034
-1.5 0.128240481035566
-1.4 0.131873953648761
-1.3 0.13538641875614
-1.2 0.138735241829921
-1.1 0.141884784548519
-1 0.14480510660951
-0.899999999999999 0.147470905590037
-0.8 0.149860707586822
-0.7 0.151956287677372
-0.6 0.15374228179497
-0.5 0.15520594594447
-0.399999999999999 0.156337020513138
-0.3 0.157127663243832
-0.199999999999999 0.157572421725376
-0.0999999999999996 0.157668223446509
0 0.1574143677351
0.100000000000001 0.156812508964536
0.2 0.155866624269659
0.300000000000001 0.154582961846561
0.4 0.152969967935319
0.5 0.151038192014387
0.600000000000001 0.148800170744627
0.7 0.146270291920867
0.800000000000001 0.143464640210581
0.9 0.140400826842321
1 0.137097805687582
1.1 0.13357567838075
1.2 0.129855491254805
1.3 0.125959026942804
1.4 0.121908593511181
1.5 0.117726813954138
1.6 0.113436418792085
1.7 0.109060044385041
1.8 0.104620039398386
1.9 0.100138281648231
2 0.095636007312504
2.1 0.0911336542274227
2.2 0.0866507207036334
2.3 0.0822056409982603
2.4 0.0778156782747968
2.5 0.0734968355783807
2.6 0.0692637850554585
2.7 0.0651298153596439
2.8 0.0611067969146947
2.9 0.0572051644553313
3 0.0534339160407533
3.1 0.0498006275371427
3.2 0.0463114813963194
3.3 0.042971308419481
3.4 0.0397836410882481
3.5 0.0367507769700113
3.6 0.0338738506601066
3.7 0.031152912708295
3.8 0.0285870139895472
3.9 0.0261742940169027
4 0.0239120717545351
4.1 0.0217969375691552
4.2 0.0198248450544079
4.3 0.0179912015727454
4.4 0.0162909564791677
4.5 0.0147186861180446
4.6 0.01326867481495
4.7 0.0119349912172054
4.8 0.0107115594670551
4.9 0.00959222481775619
5 0.00857081342334995
};
\addlegendentry{Our Ratio $C=3.72$}
\end{axis}

\end{tikzpicture}
    }
    \caption{Example of estimating the true distribution $f(\cdot)$ (dotted blue) using a surrogate $g(\cdot)$ fit from a set of current samples. The MSE approach used by prior works is shown in orange (dashed).
    While MSE \textit{looks} like a good fit, it does not directly relate to the ratio $f(x)/g(x)$ that determines the acceptance rate of rejection sampling. Our approach (solid green) has a $6\times$ higher acceptance rate over the MSE based fit used by prior works.
    }
    \label{fig:example}
\end{figure}
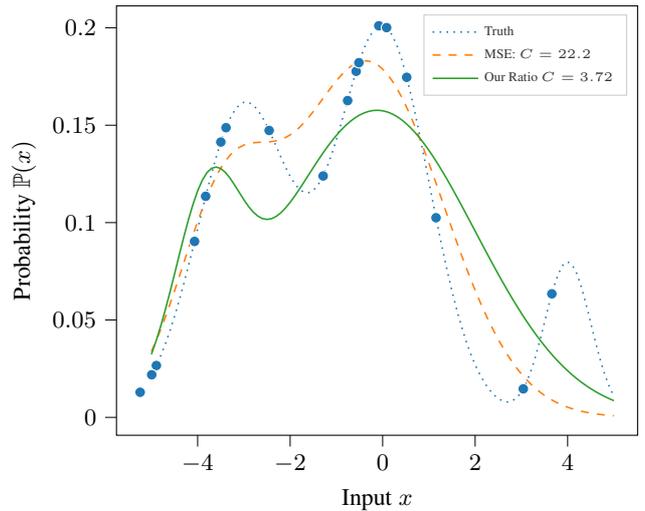

In this paper, we describe our novel Easy Rejection Sampler (ERS), where our goal is to make ERS easy for the user of rejection sampling. Rejection sampling is historically only done with differentiable functions, and so we apply modern automatic differentiation to the design of ERS. In doing so we build a system that accurately produces samples without requiring any specification of $g(\cdot)$, $C$, or any other values by the user or restrictions on $f(\cdot)$. This is obtained by leveraging work in rejection sampling that allows estimating $C$ while still producing valid samples with high probability. Since we require a proposal distribution $g(\cdot)$ that is easy to sample from, flexible, and  parameterized, we use a Gaussian Mixture Model (GMM). We exploit the differentiable approach and parameterized GMM to develop a ``refinement'' operation that maximizes the acceptance rate by directly optimizing the ratio $f(\cdot)/g(\cdot)$ instead of Mean Squared Error (MSE) based measures used by prior works. As \autoref{fig:example} shows, this can significantly reduce the true value of $C$ compared to current approaches. This refinement is run periodically as more samples are collected, allowing further minimization of $C$.

The rest of our paper is organized as follows. First we  review work related to our own in \autoref{sec:related}, including the two prior types of general purpose rejection samplers: optimization based and kernel density based. Next we will detail the design of our ERS approach in \autoref{sec:method} in three steps, presenting the primary mechanisms of ERS in the order they are used, with the final mechanism being our novel refinement based approach to proposal distributions that provides the key to higher acceptance rates over a wide array of challenging target distributions. These targets will be described in \autoref{sec:results} along with our experimental setup, followed by the results showing up to $7.3\times$ higher acceptance rates while imposing the fewest restrictions on $f(\cdot)$, and that our method is also faster by runtime estimates. These results also demonstrate the first example of obtaining improved rejection sampling runtime via the use of a GPU.

\section{Related Work} \label{sec:related}

Rejection Sampling is still widely used in multiple disciplines, often due to intrinsically low-dimensional problems or the need to perform simulations involving customized and hard to characterize integrals. This includes physics \cite{kermiche_total_2022,ozols_quantum_2012}, signal processing \cite{martino_generalized_2010}, catastrophe modeling \cite{jewson_adjusting_2020,mitchell-wallace_natural_2017}, computational finance \cite{nguyen_acceptance-rejection_2014,paskov_faster_1998}. Though these works generally use off-the-shelf rejection sampling methods or design bespoke ones, their problems are not all amenable to  current state-of-the-art samplers from machine learning, and often have high sample rates once a suitable $g(\cdot)$ and $C$ is determined. Our work will focus on challenge-problems known to produce low sample rates for current methods so that a meaningful effect is detectable.

While our method does not technically guarantee perfect sampling from the target distribution $f(\cdot)$, in all statistical tests our samples match the true distribution in every experiment. Further, approximate samples have been acceptable by users in practice\footnote{e.g., the popular \url{https://paulnorthrop.github.io/revdbayes/} uses an approximate sampler.}, and our work inherits proofs that the samples will be correct with high probability. While Hamiltonian Monte Carlo (HMC) based sampling also uses a gradient and can handle much higher dimensions than rejection sampling, it also is not as `turn key` and can require non-trivial work to ensure sample quality. Our study is focused only on rejection sampling and its standard uses: differentiable functions with $d \leq 3$ dimensions ~\cite{kermiche_total_2022,ozols_quantum_2012,jewson_adjusting_2020,mitchell-wallace_natural_2017,nguyen_acceptance-rejection_2014,paskov_faster_1998,Jul2021}. 

There is a broad family of \textit{adaptive} rejection sampling methods first proposed by \cite{Genz1992}, where the sampling acceptance rate improves for larger total number of samples $n$. This seminal work has spawned many extensions, but which generally require the specification of $g(\cdot)$ and $C$ (though they may ease the process), and more restrictive limitations on $f(\cdot)$ \cite{Gorur2011,Casella2004,Evans1998,Hormann1995}. Our work is instead focused on more general purpose rejection sampling techniques with weak or limited assumptions on $f(\cdot)$ that require little work to apply to new problems. 

There are two different families of general purpose rejection sampling algorithms available today. The first family can be described as an optimization \& sampling strategy that was first proposed with the OS* algorithm \cite{Jul2021}. Their work introduced the use of rejected samples to improve the proposal $g(\cdot)$. The latter A* introduced a gumbel trick to further improve efficacy \cite{NIPS2014_309fee4e}. However these methods require additional derivations to be done by the user, requiring a function $i(\cdot)$ and $o(\cdot)$ such that $f(x) \propto \exp (i(x)+o(x))$. Our ERS also intermixes optimization with sampling, but uses gradient based approaches. In contrast ERS will not require any proposals or bounds to be specified by the user that OS* and A* require. 

Pliable Rejection Sampling (PRS) \cite{10.5555/3045390.3045614} introduced the second strategy of using a Kernel Density Estimator (KDE) to estimate $g(\cdot)$ from the samples, a strategy refined by Nearest Neighbour Adaptive Rejection Sampling algorithm (NNARS) \cite{pmlr-v98-achddou19a} to produce a near-optimal sampler under certain assumptions. However, both methods require bounded support to estimate the KDE and have non-trivial parameters to estimate analytically, or via cross-validation. The later is particularly challenging as it requires pre-existing samples, which significantly complicates practical usage. In contrast we impose no constraints on $f(\cdot)$ and require no other items to be specified by the user, and we use a GMM so that we may modify a reasonable number of parameters via gradient descent. Though PRS and NNARS are most similar to ERS in terms of $g(\cdot)$, they optimize a form of MSE to approximate $f(\cdot)$ where our novel ratio optimization provides considerably high acceptance rates.

We make note that the considerable requirements to use modern rejection samplers create reproducibility risks. The latter cited works often depend on transformations of the sampling distribution in order to convert unbounded support into finite support (i.e, changing the domain $\mathcal{X}$ from $[0, \infty]$ to $[0, 1]$) but do not specify what transformation is used, and do not specify many derivations.  NNARS relies on a Holder constant $H$ and associated $s \in [0, 1]$ such that $|f(x)-f(y)| \leq H\|x-y\|_{\infty}^{s}, \forall x, y \in \mathcal{X}$, but the $H$ values are no stated for any experiments\footnote{Code for the method is available, not the experiments, and we were unable to replicate their results in our attempt.}. Similarly, PRS requires a cross-validation procedure with no further details. This missing information has been identified by many prior works as a significant barrier to replication \cite{Dodge2019,Gundersen2018,Musgrave2020,Raff2019_quantify_repro,Raff2020c,10.1145/3589806.3600041,10.1145/3589806.3600042}. For these reasons, we use the results as presented in these prior works, but note an additional benefit of our approach: by requiring no hyper-parameters, transformations or other constraints to be specified, the ability to reproduce works leveraging ERS is increased. Further, we provide source code at \url{https://github.com/NeuromorphicComputationResearchProgram/EasyRejectionSampling}. 

\section{Method} \label{sec:method}

Our primary goal is to design a mechanism by which a user is not required to specify anything other than the function $f(\cdot)$, and optionally a domain $\mathcal{X}$ that they wish to draw samples from $x \sim f(\cdot) \in \mathcal{X}$. This is a more challenging problem than that addressed by prior methods, as the user, by specifying both the surrogate function $g(\cdot)$ and the ratio $C = \sup_{z \in \mathcal{X}} \frac{f(z)}{g(z)}$, is imposing a prior knowledge to the sampling method. 

To build our approach, we must define a surrogate function $g(\cdot)$ and somehow determine the maximal ratio $C$ during runtime. We tackle the latter by leveraging the little known method of ``Empirical Supremum'' based sampling \cite{Caffo2002}. 
Caffo et al.
showed that by initializing an estimate of the supremum $\hat{C}_t = 1$, and then adjusting it to $\hat{C}_{t+1} = \max\left(\hat{C}_t, \frac{f(x_{t+1})}{g(x_{t+1})}\right)$, 
that there will be a finite number of errors when the support $\mathcal{X}$ is discrete,
with a similar result for unbounded support. However, \cite{Caffo2002} still required the specification of $g(\cdot)$. We note that the Empirical Supremum was found to converge in under 100 iterations. We leverage this by performing a minimum of 500 samples per iteration, using the maximum across all samples before deciding an accept/reject decision. In practice we find that this results in  convergence within $0.001$ in one iteration, resulting in no difference between using $\hat{C}$ and the true $C$ for our experiments. 

Understanding that we will safely estimate $C$ as the sampling runs allows us to specify the overall strategy of our approach. The description is subdivided into three critical design choices. Together, they will form our ``Easy Rejection Sampling'' approach.
\begin{enumerate}
    \item We will use a batched iterated sampling procedure to alter the candidate distribution $g(\cdot)$ via a Gaussian Mixture Model after successive rounds of sampling. 
    \item We will specify a simple gradient based strategy for selecting an initial candidate distribution $g(\cdot)$ with few evaluations of $f(x)$.
    \item We develop an approach to \textit{refining} the distribution $g(\cdot)$ to \textit{maximize acceptance rate}, where prior work has instead focused on maximizing the overall similarity of the distributions. 
\end{enumerate}

We note that the third refinement step, where we develop a loss function that targets the ratio $f(\cdot)/g(\cdot)$ directly, is critical to the success of our method. It is placed at the end due to readability, and it is the last step of each iteration. This is necessary as we use both accepted and rejected samples to inform the refinement, so running the refinement regularly can further improve the results by reducing $C$ further.

Careful design allows us to perform significantly better than the theoretically optimal (in only a min-max sense\footnote{Their proof is with respect only to families of the  H\"{o}lder density, in worst-case situations. This does not inform a limit on the rejection rate for non-worst-case densities within that family, let alone those beyond it. For example, densities with infinite support already are outside the scope of these proofs. This is not to diminish the value of their theoretical contributions, but to elaborate that significant room for improvement is still possible.}) NNARS algorithm while simultaneously requiring fewer constraints and no extra information from the user. This shows that the constant factors in NNARS are non-trivial, and full understanding of the limits of adaptive rejection sampling is not yet known.   Our method further inherits the correctness results from  \cite{Caffo2002}, and additional empirical tests show no detectable difference in our sample quality from the true distribution.

Each of the following sub-sections match our stated design choices, and will detail the relevant approach with the justification for why the approach was taken. We note that, as a matter of engineering, the below procedures currently include hard-coded hyper parameters that we do not tune. Their purpose is simple and intuitive (e.g., increase a value by some small amount so we are not comparing to a worst case), and do not require complex math like the Holder constant used by NNARS. In extended testing we found that these coefficients did not have a meaningful impact on the results. 

Our implementation is done in JAX~\cite{jax2018github} and works on $\log(f(\cdot))$ and $\log(g(\cdot))$ in order to be numerically stable. Our description will remain at the $f(\cdot)$ and $g(\cdot)$ level for clarity. 

\subsection{Easy Rejection Sampling Algorithm}

The overall procedure of our approach can be succinctly described as fitting a GMM to the current data as the proposal distribution $g(x)$, sampling, and then re-fitting the proposal $g(x)$ at various intervals. Our use of vectorized operations makes the multiple fittings of $g(x)$ computationally reasonable, and when factoring out differences in acceptance rate is $\geq 2\times$ faster on CPU, and shows the first speedup for a rejection sampler on a GPU at $4\times$. 

Our approach is shown in pseudo-code in Algorithm \ref{algo:ers}, where $A$ and $R$ contain the set of currently accepted and rejected samples, with their evaluations against $f(\cdot)$ cached for posterity. The $K$ current mean, covariance, and weights of the GMM are $\bm{\mu}_{1, \ldots, K}, \Sigma_{1, \ldots, K},$ and $\bm{w}$ respectively; the subscript ${1, \ldots, K}$ will be dropped if the value of $K$ is not changing. Our two primary operations are fitting a GMM (initialized with $k$-means++ \cite{Arthur2007,Raff_2021_kmeans}), and a ``refine'' step that we discuss in \autoref{sec:refinement}.

\begin{algorithm}[!h]
\caption{Easy Rejection Sampling} \label{algo:ers}
\begin{algorithmic}[1]
\Require Target function to sample $f(\cdot)$, limited to a domain $\mathcal{X}$. A target number of samples to draw $T$ 
    \State $\bm{\mu}_{1, \ldots, K}, \Sigma_{1, \ldots, K}, \bm{w} \gets $ \Call{Initialize}{$f$, $\mathcal{X}$} \Comment{$g(x)$ is short for $\sum_{i=1}^K w_i \cdot \mathcal{N}_\mathcal{X}(\bm{\mu}_i, \Sigma_i) $, see \autoref{algo:init}. Note: $K$ is the number of initial means determined by the Initialize function. }
    \State $A \gets []$ \Comment{Accepted samples}
    \State $R \gets []$ \Comment{Rejected samples}
    \State $\hat{C} \gets -\infty$, $C_\mathit{low} \gets \infty$
    \State refine $\gets$ \textit{False}
    \While{$|A| < T$} 
         \State $X^{n \cdot \log(K+1), T} \sim g(\cdot)$ \Comment{Draw $n \cdot \log(K+1)$ candidate samples}
         \State $\tilde{C} \gets \max_i\left( \frac{f(x_i)}{g(x_i)} \right)$ \Comment{Batch supremum}
         \State refine $\gets \tilde{C} > \hat{C} \lor \tilde{C} > C_{\mathit{low}}$
         \State $\hat{C} \gets \max(\hat{C}, \tilde{C})$ \Comment{Empirical supremum}
         \State $C_{\mathit{low}} \gets \min(C_{\mathit{low}}\cdot1.05, \tilde{C})$
         \State $u_i \sim \mathcal{U}(0, 1)$
         \State Accept to $A$ all $u_i \leq \frac{f(x_i)}{\hat{C} \cdot g(x_i)}$, all rejected to $R$ \Comment{$f(x_i)$ is cached for \textit{all} samples}
        \State $\bm{\tilde{\mu}}, \tilde{\Sigma}, \bm{\tilde{w}} \gets \bm{\mu}, \Sigma, \bm{w} $
         \If{Last GMM call was more than $|A|/1.5$ acceptances ago or $\log(|A|) > 2\cdot K$}
            \State $\bm{\mu}_{1, \ldots, K'}, \Sigma_{1, \ldots, K'}, \bm{w} \gets $ fit GMM to $A \cup R$ with $K' = \min(\log_2(|A|), |A|/(d \cdot 15))$ clusters
            \State refine $\gets$ \textit{True}
         \EndIf
         \If{refine}
            \State $\bm{\mu}, \Sigma, \bm{w} \gets $ \Call{Refine}{$\bm{\mu}$, $\Sigma$, $\bm{w}$, $A$, $R$} \Comment{Refine initial and GMM, taking the best (lowest $\bar{C}$), see \autoref{algo:refine}}
            \State $\bar{C} \gets \max_i\left(\frac{f(\bm{x}_i)}{g(\bm{x}_i)}\right), \quad \forall \bm{x}_i \in A \cup R$ 
            \If{$\bar{C} \leq \hat{C}$}
                \State $\hat{C} \gets \bar{C}$\Comment{We keep current refinement}
            \Else
                \State $ \bm{\mu}, \Sigma, \bm{w} \gets \bm{\tilde{\mu}}, \tilde{\Sigma}, \bm{\tilde{w}} $ 
                \Comment{We reject current refinement, and revert to previous model.}
            \EndIf
            \State refine $\gets$ \textit{False}
         \EndIf
    \EndWhile
\end{algorithmic}
\end{algorithm}

We sample $n=500$ items at a time, and increase the value multiplicatively with the log of the number of components $K$ in the mixture to ensure sampling is computationally effective. We note on line 4 that we initialize the empirical supremum $\hat{C}$ to $-\infty$ instead of 1, because it allows our approach to work with both unnormalized $f(\cdot)$ and unnormalized $g(\cdot)$. The former reduces the complexity of the implementation, especially for intractable normalization terms that would require estimation. The latter helps with bounded support $\mathcal{X}$, avoiding more costly and complex evaluations of $g(\cdot)$. This is possible because $\hat{C}$ is simply determining the maximum observed ratio, and so any missing normalizing terms in either function roll into $\hat{C}$ multiplicatively (i.e, $\frac{1}{\hat{C}} \cdot \frac{f(\cdot)/z_1}{g(\cdot)/z_2} = \frac{z_2}{z_1 \hat{C}} \frac{f(\cdot)}{g(\cdot)}$ ).  

Lines 7-13 perform the rejection sampling, where $\tilde{C}$ is the \textit{batch} supremum, which in most all cases converges (or closely approaches convergence) in a single iteration.  We also keep track of $C_{\mathit{low}}$, a running estimate of the lowest batch supremum seen inflated by 5\% per round. If the current batch supremum is above $C_{\mathit{low}}$, we flag the model for refinement. While violating $C_{\mathit{low}}$ does not change the empirical supremum $\hat{C}$, it is an indicator that we may have acquired new---relatively harder---samples that will benefit the refinement optimization later.  A refinement is also done if the bound $\hat{C}$ occurs, but this is rare and usually occurs due to samples occurring out in the tail of the distribution (e.g., $f(x)/g(x)$ \textit{could} be larger for an $x$ such that $\max(f(x), g(x)) \lessapprox
1/n$). 

A \textit{potential} GMM is fit whenever the number of accepted samples has increased by 50\%, to avoid excessive training of GMMs. We do not require fitting an accurate distribution, but rather one that can be refined to a maximal acceptance rate, and do not require finding the ``optimal'' number of mixtures. As such we use the standard expectation maximization approach with diagonal covariance, and simply make the number of  mixtures in the new GMM grow logarithmically with the accepted sample size. The diagonal covariance is used because it is faster to sample from by a significant margin, especially when a compact support $\mathcal{X}$ is given where samples can be drawn in $\mathcal{O}(1)$ time, as the general case for an arbitrary covariance is highly non-trivial~\cite{Botev2017}.  We use $\mathcal{N}_\mathcal{X}$ to denote the normal distribution truncated to the domain of $\mathcal{X}$, which is user specified (e.g., a function in the domain $[0, \infty]$). 

\textit{Why not full rank covariance?} Using a full-rank $\Sigma$ is problematic when we have compact support $\mathcal{X}$ (i.e., $x \in [0, \infty]$) because there is no closed form method of sampling from the truncated distribution $\mathcal{N}_\mathcal{X}$, which thus requires its own sampling scheme to drawn from~\cite{Wilhelm2010}. This is extremely expensive, and we found that the Gaussian clusters that occur at the edge of the support $\mathcal{X}$ are very challenging to sample from and would increase run-time by $10,000\times$ just due to the cost of sampling from the resulting $g(\cdot)$. Thus diagonal $\Sigma$ is preferred due to exact and fast sampling from it's truncated distribution. 

A key optimization is that the GMM is fit to both accepted and rejected samples, by weighting each data point $x_i$ by its true evaluation $f(x_i)$. This is essentially free, as we cache all calls to $f(\cdot)$, and requires at most doubling the memory use. Heuristically, we multiply the weight of accepted samples by a factor of 10 to reflect their greater importance to the underlying distribution, which allows our method to still work even when initial sampling rates are low. These factors all occur in  lines 15-17. 

Finally, we perform the refinement on lines 18-24. This is a non-convex problem, and so does not always succeed. If a GMM was attempted on lines 15-17, both the current model and the candidate GMM will be refined. We can compute the empirical supremum using the cached $f(\cdot)$ values again, to determine if the refinement has provided a new distribution $g(\cdot)$ that is better than prior solutions. 

We can formally show the proofs of \cite{Caffo2002} still apply, using their notation:
\begin{theorem} \autoref{algo:ers}, given a fixed $g(\cdot)$, and a sequence of $i$ samples draw thus far, converges to the same or better (fewer false samples) solution as \cite{Caffo2002}, and thus retains $\mathcal{O}(i^{-1})$ convergence rate of correctness. 
\end{theorem}
\begin{proof}
Let $\widetilde{\tau}_{i}=\min \left\{j \in \mathbb{N} \mid U_{i j} \leqslant \frac{f\left(X_{i j}\right)}{\widehat{C}_{i} g\left(X_{i j}\right)}\right\}$ define the sequence of samples drawn from $g(\cdot)$ that define the index of the $i$'th sample $X_{i, \widetilde{\tau}_i}$ sampled by Empirical Rejection Sampling, and $\tau_i$ the result from using the true supremum $C$. If samples are selected $B$ at a time by \autoref{algo:ers}, then the sampled index $\tau^B_i  = \min \left\{j \in \mathbb{N} \mid U_{i j} \leqslant \frac{f\left(X_{i j}\right)}{\widehat{C}_{i + (B-i \mod B)} g\left(X_{i j}\right)}\right\}$. By definition $\widehat{C}_{i} \leq \widehat{C}_{i+1}$, and so a simple recurrence shows that $\widehat{C}_i \leq \widehat{C}_{i + (B-i \mod B)}$. Thus it must be the case that $\tau^B_i \leq \widetilde{\tau_i}$. Since $\widetilde{\tau_i}$ controls the acceptance of samples and forms the proof of correctness for \cite{Caffo2002}, then \autoref{algo:ers} also satisfies the proof for a fixed $g(\cdot)$. 
\end{proof}

The proof that you can alter $g(\cdot)$ is of the same form by ``restarting'' the sequence when $g(\cdot)$ changes, and using the previous $X_{i,j}$ values to pick an initial $\widehat{C}_1$ that is valid (true by definition, as its the maximum ratio of $f(\cdot)/g(\cdot)$ observed so far), and beginning a new convergence of rate $\mathcal{O}(i^{-1})$ at the warm-started solution $\widehat{C}_1$. 

This proves that our \autoref{algo:ers} will converge to a solution of reasonable quality, it \textit{does not guarantee that no erroneous samples will occur}. The probability of this can be described, but not easily quantified, as the likelihood the $i$'th sample $x_i$ being drawn incorrectly is the situation that $\frac{f(x_i)}{C g(x_i)} < u_i < \frac{f(x_i)}{ \widehat{C}_{i + (B-i \mod B)} g(x_i)}$, which simplifies to answering the probability that $\mathbb{P}\left(C < \frac{ \widehat{C}_{i + (B-i \mod B)} f(x_i)}{f(x_i) -  \widehat{C}_{i + (B-i \mod B)} g(x_i) u_i} \right)$. This is hard to quantify due to the joint dependence on a future iteration's estimate of $\hat{C}$ (because we get samples in batches), the curvature of $f(\cdot)$ and $g(\cdot)$, and the uniform random value of $u_i$. We instead use the final value of $\hat{C}$ to check that empirically, we do not appear to have falsely accepted any samples. All of our experiments passed this test. 

\subsection{Initial Proposal Distribution} \label{sec:init}

Now that we have specified the overall iterative strategy of our rejection sampler, we specify how the initial distribution $g(\cdot)$ is chosen. Congruent with standard practice, we aim for our initial proposal to be wider than the underlying distribution $f(\cdot)$, and let subsequent iterations of our algorithm narrow the proposal once samples have been obtained. Note that we count all $f(\cdot)$ evaluations in this stage against the total number of calls for computing sample acceptance rate, so it is important that we find a reasonable choice with a limited number of $f(\cdot)$ evaluations.

The procedure is outlined in \autoref{algo:init}, and contains three primary steps: 1) handling finite support, 2) apparently unimodal distributions, 3) multi-modal distributions. 

First, if the distribution was indicated to have a bounded support, we simply set $\bm{\mu}$ to have the center of the min/max bounds, and set the covariance to be wide enough to cover the entire space. This occurs on lines 3-7, and handles the finite support case. 
On lines 8-9 we have potentially infinite locations for the distribution, and so use random sampling to find a location that has a non-zero probability. This provides greater flexibility for our approach. 

Once an initial point is selected, on lines 10-13 we sample a small number of points in the space around the initial point, and then use Stochastic Gradient Descent (SGD) to maximize the value of $f(\cdot)$. Because we are trying to empirically find some number of modes of $f(\cdot)$, we use fast converging FISTA~\cite{BECK2003167} to quickly reach local maxima as implented in jaxopt~\cite{Blondel2021a}. These modes are used to determine whether or not the distribution is multi-modal. 

\begin{algorithm}[H]
\caption{Initialization of first proposal distribution $g(x) = \sum_{i=1}^K w_i \cdot \mathcal{N}(x | \mu_i, \Sigma_i) $}  \label{algo:init}
\begin{algorithmic}[1]
\Procedure{Initialize}{$f$, $\mathcal{X}$}       
    \State $x \gets \vec{0}$ 
    \If{$\mathcal{X}$ is compact} %
        \State $\mu_1 \gets $\Call{Center}{$\mathcal{X}$} \Comment{Place $g(x)$ at the center.}
        \State $\Sigma_1 \gets $\Call{Range}{$\mathcal{X}$}/3 \Comment{Set radius to cover the whole space with 3 $\sigma$}
        \State $w \gets [1]$ \Comment{Uni-modal}
        \State \Return $\mu_1, \Sigma_1, \bm{w}$
    \EndIf

    \While{$f(x) = 0$}  \Comment{Maybe discontinuous?}
        \State $x \sim \mathcal{N}(0, I)$
    \EndWhile  
    \For{$i \in [1, d+3]$}
        \State $\mathit{modes}_i \gets x + \sim \mathcal{N}_\mathcal{X}(0, I)$
        \State $\mathit{modes}_{d+4} \gets x$
    \EndFor
    \State Run SGD on $\mathit{modes}$ to minimize $\sum_i -f(\mathit{modes}_i)$
    
    \If{Covariance of $\mathit{modes}_{\ldots} < \epsilon$  } \Comment{We need to estimate a covariance}
        \State $\mu_1 \gets \frac{1}{K}\sum_i^K \mathit{modes}_i$
        \For{$i \in [1, d\cdot 2+10]$}
            \State $\mathit{spread}_i \gets x + \sim \mathcal{N}_\mathcal{X}(0, I)$
        \EndFor
        \State Run SGA on $\sum_{i} \log\left(f(x)-\mathit{spread}_i-5\right)^2$
        \State $\Sigma_1 \gets $\Call{cov}{$\mathit{spread}_{\ldots}$}
        \State $w \gets [1]$
        \State \Return $\mu_1, \Sigma_1, \bm{w}$
    \Else \Comment{Estimate multi-modal coverage}
        \State Select $K$ values $\mu_1, \mu_2, \ldots, \mu_K$ from $\mathit{modes}$ via k-farthest selection, stopping when $\|\mu_{K} - \mu_{K+1}\|_2 < \epsilon$
        \State Set $\Sigma_i \gets I\cdot(\max_{j,z} \|\mu_j-\mu_z\|_2/K), \quad \forall i \in [1, K]$
        \State $\bm{w} \gets \vec{1}/K$
        \State \Return $\mu_{\ldots}, \Sigma_{\ldots}, \bm{w}$
    \EndIf
\EndProcedure
\end{algorithmic}
\end{algorithm}

We check for unimodality by checking the covariance of the found modes, and if smaller than some $\epsilon$, that means all points are on top of each other and thus the distribution appears unimodal (if that was an erroneous decision, it will eventually be corrected by the sampling and GMM refits). To estimate a covariance matrix, we again use SGD to find points that have a different in log probability of $-5$ lower, which is many orders of magnitude smaller. The covariance is estimated from these points, which---being overly far from the mode of the distribution---will have heavier tails than $f(\cdot)$ and thus allow good sampling coverage. 

If the covariance is non-zero, we begin selecting the $2$ farthest pairwise points, and iteratively looking at the next point that is farthest from the current set until the distance becomes $< \epsilon$, indicating that we are not selecting any new modes. This gives us a set of $K$ modes, and we use the maximum pairwise distance divided by the number of modes as the covariance for all modes. This again provides an overestimate of the true covariance and thus helps ensure coverage. 

Combined, these give us a strategy for selecting the initial distribution $g(\cdot)$ that will quickly identify empirical maximum ratios $f(\cdot)/g(\cdot)$, and allow fast convergence of our methods.

\subsection{Refinement} \label{sec:refinement}

The final, and most significant improvement of our method is the refinement operation, where we alter the GMM to improve the acceptance rate of samples drawn. Our key insight is that the kernel density estimates used by NNARS and PRS are attempting to minimize losses of the form $\int_\mathcal{X} \left|f(x) -  g(x)\right|^p dx$. This is intuitive: the closer $f(\cdot)$ and $g(\cdot)$ are, the higher the acceptance rate will be. Yet, this is not what the core rejection sampling procedure addresses. Instead we can seek to refine the initial proposal model $g(\cdot)$ to minimize $\sup_{z \in \mathcal{X}} \frac{f(z)}{g(z)}$ in a direct fashion, such that we should expect high acceptance rates. 

\begin{algorithm}[H]
\caption{Refinement of $\bm{\mu}_{1, \ldots, K}, \Sigma_{1, \ldots, K}, \bm{w}$ using accepted and rejected samples $A$ and $R$} \label{algo:refine}
\begin{algorithmic}[1]
\Procedure{$\ell$}{$\bm{\mu}$, $\Sigma$, $\bm{w}$, $A$, $R$ }  
    \State $\alpha_i \gets \log(f(\bm{x}_i)/g(\bm{x}_i)), \quad \forall \bm{x}_i \in A$
    \State \Return $\langle$\Call{softmax}{$\bm{\alpha}$}, $\bm{\alpha}\rangle$
\EndProcedure
\While{Max iterations not reached}
    \State Alter $\bm{\mu}$, $\Sigma$, and $\bm{W}$ using $\frac{\partial \ell}{\partial \bm{\mu}}$, $\frac{\partial \ell}{\partial \Sigma}$, and  $\frac{\partial \ell}{\partial \bm{w}}$ via automatic differentiation 
    \If{Current solution is $\leq$ than $\hat{C}$}
        \State \Return $\bm{\mu}$, $\Sigma$, $\bm{w}$
    \EndIf
\EndWhile
\State \Return original  $\bm{\mu}$, $\Sigma$, $\bm{w}$
\end{algorithmic}
\end{algorithm}

Our approach to doing so is simple and detailed in \autoref{algo:refine}. We use the existing samples to compute the log ratio $\alpha_i = \log(f(\bm{x}_i))-\log(g(\bm{x}_i))$, giving us a vector of empirical results. While we could select the maximal $\alpha_i$ to use as the loss to minimize, this is undesirable, as many $\alpha_i$ may be large and the optimization will take longer (especially when many mixture components $K$ exist, using the maximum will generally influence only one of $K$ components). Instead we use the approximate maximum computed by $\text{SOFTMAX}(\bm{\alpha})^\top \bm{\alpha}$, which allows a better behaved gradient to impact multiple components $K$ in proportion to their log-ratios. 
Optimizing $\text{SOFTMAX}(\bm{\alpha})^\top \bm{\alpha}$ still directly tackles the sampling ratio $f(\cdot)/g(\cdot)$, and so results in higher acceptance rates via lower $C$ values. Crucially this requires no further calls to $f(\cdot)$ because it is performed only on the current samples (accepted and rejected) which have already cached $f(\cdot)$ values.  

This approach may be seen as an adaption of the strategy used for gumbel softmax sampling trick of \cite{Jang_Gu_Poole_2017}, but instead of reparameterizing a target function we are creating a biased approximation that is advantageous in practice. 
In our context, the ``correct'' target function is to replace line 3 of algorithm 3 with $\max_{\forall i} \alpha_i$. This uses the maximal value (of $\frac{f(\cdot)}{g(\cdot)}$), but also means that $\forall j \neq i$, the gradient through $\alpha_j$ is exactly equal to 0 (i.e., the max operation returns a non-zero gradient only to the index that was selected).
This is problematic when the ratio $\frac{f(\cdot)}{g(\cdot)}$ is multimodal, because only one mode of $g(\cdot)$ will get a meaningful gradient because only one point $\boldsymbol{x}_i$ contributed to the final loss calculation (via $\alpha_i$). By using the softmax each mode of $\frac{f(\cdot)}{g(\cdot)}$ we get a gradient proportional to its scale, and so progress over the whole domain is made instead of just one location. Thus we replace the target function of interest with a \textit{biased} approximation (i.e., Softmax($\boldsymbol{\alpha})^\top \boldsymbol{\alpha}$ $\neq \max_i \alpha_i$) because it results in better learning behavior. In comparison, \cite{Jang_Gu_Poole_2017} need a well-defined gradient through a stochastic function, and so require the softmax for their approach to define that gradient that is unbiased. But we both use the softmax to get a well-behaved gradient. 

Because this is a non-convex optimization problem, we use the AdaBelief\cite{zhuang2020adabelief} optimizer with an initial learning rate of 0.1, and perform 800 total gradient update steps per refinement call. We check the quality of the solution at 100, 200, 400, and 800 steps and select the best (lowest maximum ratio). Efficacy could be improved by checking at every iteration, but is not advantageous from a computational perspective due to the interpreter overheads of Python. 

\section{Experiments and Results} \label{sec:results}

Having defined our approach, we first specify our experiments, followed by results showing that ERS is competitive with or significantly better than prior methods while simultaneously requiring less information from the user. We will use several standard benchmarks used for rejection sampling problems that each exercise a different kind of challenge for which a sampler may suffer low acceptance rates. The primary measure of interest is the acceptance rate, which indicates what percentage of samples $g(\cdot)$ are accepted divided by the total number of evaluations of $f(\cdot)$. This means the gradient descent operations used in \autoref{sec:init} reduce the total rate of ESR. The gradient evaluations in \autoref{sec:refinement} do not count because the values of $f(\bm{x}_i)$ are cached during the initial sample acceptance/rejection evaluation, so no additional evaluations of $f(\cdot)$ are needed. Consistent with prior work, we use $n=10^5$ target accepted samples and 10 runs to compute the mean and standard deviation of results. In all cases we report the NNARS, PRS, A*, and OS* results from prior works as we were unable to replicate their success. 

Three primary and challenging problems are used as the target function $f(\cdot)$. The first is the ``peakiness'' problem proposed by \cite{NIPS2014_309fee4e} in \autoref{eq:peakiness}, where $a$ controls how ``peaky'' the distribution is and the domain $\mathcal{X} \in (0, \infty)$. As $a \to \infty$ the peakiness gets higher, making sampling more difficult. 

\begin{equation} \label{eq:peakiness}
    f(x) \propto \frac{e^{-x}}{(1+x)^{a}}
\end{equation}

The next problem tests the impact of scaling the dimension size of the problem and is given for a general $d$-dimensional distribution in \autoref{eq:scaling} as proposed by \cite{10.5555/3045390.3045614}. Here the support is compact with $\mathcal{X} \in [0, 1]$. We note that this distribution is highly multi-modal, making it especially challenging. Most rejection sampling focuses on one-dimensional problems due to the difficulty of specifying a useful $g(\cdot)$ and $C$. 

\begin{equation} \label{eq:scaling}
f(\bm{x}) \propto \prod_{i=1}^d \left(1+\sin \left(4 \pi x_i -\frac{\pi}{2}\right)\right)
\end{equation}

The last test we consider is the ``clutter'' problem first posed by \cite{Minka2001}, and is given in \autoref{eq:clutter}. Here $\theta$ indicates a set of centers that are selected from 10 points spaced uniformly in the range of $[-5, -3]$ and $[2, 4]$. This creates a distribution with two very strong peaks that are separated by a far distance, and has been a challenging sampling distribution for over two decades \cite{Minka2001}. 

\begin{multline} 
\label{eq:clutter}
    f(\boldsymbol{x})\propto \prod_{i=1}^N r \left((2 \pi)^{-d/2} e^{-(\bm{x} - \theta_i)^2/2}\right) \\ +  (1-r) \left((2 \pi)^{-d/2} e^{-(\bm{x}^2/100^2)/2} 100^{-d}\right)
\end{multline}

For all tests where $d \leq 2$ we ran a two-sample Kolmogorov–Smirnov (KS) test comparing ERS's samples with those of $A*$ as implemented in the original code, or with the true distribution. In no case was a difference in ERS's samples and the target distribution detected. Because the KS test is not well defined for $d > 2$, in these situations we used a  two-sample Carmer test, which also found no significant difference. 

We note that reproducibility of NNARS and PRS is limited. For NNARS and PRS \autoref{eq:peakiness} results are given, though the distribution is not compact -- and no transformation $[0, \infty] \to [0, 1]$ is specified. Similarly for PRS results on the clutter task \autoref{eq:clutter} are given, and were obtained by artificially clipping the distribution to a range that contained all samples\footnote{This was indicated by the author when asked over email, though they do not remember the clipping value. We appreciate their responsiveness and valuable information that helped us determine replicability.}. Both methods have hyper-parameters that are not specified and would be altered by a chosen transformation, further complicating our replication of their results. Since NNARS presents the state-of-the-art results, we use their reported results as our comparison numbers though we are unable to reproduce them.

\subsection{Empirical Acceptance Rates}

Results will be presented in the same order as the problems were specified. Results also present a ``Simple Rejection Sampling'' (SRS) baseline of manually specifying $g(\cdot)$ and $C$ as reported by prior work.

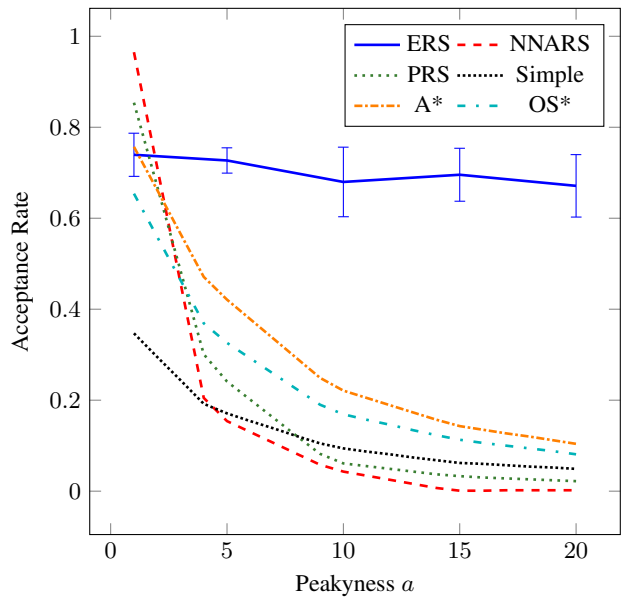
\begin{figure}[!h]
    \begin{tikzpicture}
\begin{axis}[
    error bars/y dir=both,
    xlabel=Peakyness $a$,
    ylabel=Acceptance Rate,
    legend pos=north east,
    legend columns=2,
    height=0.99\columnwidth,
    width=0.99\columnwidth,
    cycle list name=MyCyclelist
    ]
    
    \addplot+[line width=1.0pt, mark=none, solid] table [x=Dimension, y=ERS,y dir=normal,col sep=comma,y error=ERS_std] {data/peakness.csv};
    \addlegendentry{ERS}
    
    \addplot+[line width=1.0pt, mark=none, dashed] table [x=Dimension, y=NNARS, col sep=comma] {data/peakness.csv};
    \addlegendentry{NNARS}
    
    \addplot+[line width=1.0pt, mark=none, dotted] table [x=Dimension, y=PRS, col sep=comma] {data/peakness.csv};
    \addlegendentry{PRS}
    
    \addplot+[line width=1.0pt, mark=none, densely dotted] table [x=Dimension, y=Simple, col sep=comma] {data/peakness.csv};
    \addlegendentry{Simple}
    
    \addplot+[line width=1.0pt, mark=none, densely dashdotted] table [x=Dimension, y=AStar, col sep=comma] {data/peakness.csv};
    \addlegendentry{A*}
    
    \addplot+[line width=1.0pt, mark=none, loosely dashdotted] table [x=Dimension, y=OSStar, col sep=comma] {data/peakness.csv};
    \addlegendentry{OS*}
\end{axis}
\end{tikzpicture}
\caption{Results on the peakiness problem of \autoref{eq:peakiness}, where $a=1$ indicates minimal peakiness and easier samping, and $a=20$ is higher peakiness and more challenging to sample from. NNARS and PRS only perform better for $a=1$, where all approaches perform well. Our ERS suffers only a 1\% point drop in mean acceptance rate as $a$ increases at each step, where all other approaches degrade quickly. 
}
\label{fig:peaky}
\end{figure}
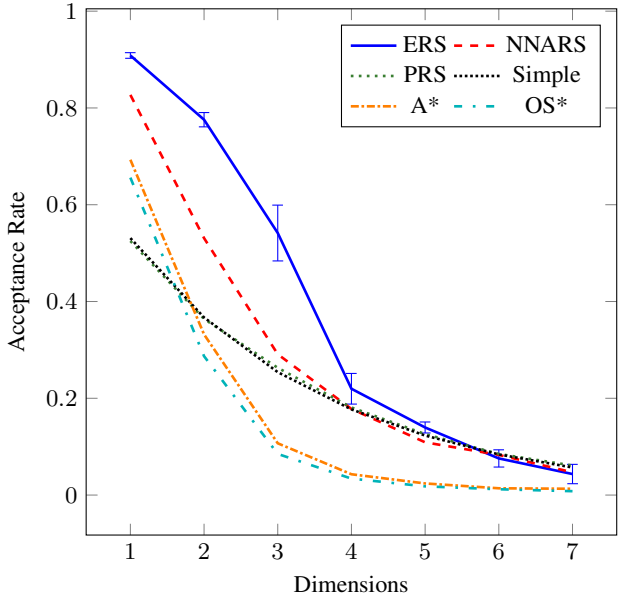
\begin{figure}
\begin{tikzpicture}
\begin{axis}[
    error bars/y dir=both,
    xlabel=Dimensions,
    ylabel=Acceptance Rate,
    legend pos=north east,
    legend columns=2,
    height=0.99\columnwidth,
    width=0.99\columnwidth,
    cycle list name=MyCyclelist
    ]
    
    \addplot+[line width=1.0pt, mark=none, solid] table [x=Dimension, y=ERS,y dir=normal,y error=ERS_std, col sep=comma] {data/dimension_scaling.csv};
    \addlegendentry{ERS}
    
    \addplot+[line width=1.0pt, mark=none, dashed] table [x=Dimension, y=NNARS, col sep=comma] {data/dimension_scaling.csv};
    \addlegendentry{NNARS}
    
    \addplot+[line width=1.0pt, mark=none, dotted] table [x=Dimension, y=PRS, col sep=comma] {data/dimension_scaling.csv};
    \addlegendentry{PRS}
    
    \addplot+[line width=1.0pt, mark=none, densely dotted] table [x=Dimension, y=Simple, col sep=comma] {data/dimension_scaling.csv};
    \addlegendentry{Simple}

    \addplot+[line width=1.0pt, mark=none, densely dashdotted] table [x=Dimension, y=AStar, col sep=comma] {data/dimension_scaling.csv};
    \addlegendentry{A*}
    
    \addplot+[line width=1.0pt, mark=none, loosely dashdotted] table [x=Dimension, y=OSStar, col sep=comma] {data/dimension_scaling.csv};
    \addlegendentry{OS*}

\end{axis}
\end{tikzpicture}
\caption{Results for \autoref{eq:scaling} as the number of dimensions $d$ increases. In this case ERS provides superior acceptance rates for low dimensionality, but suffers from the curse of dimensional similarly to prior KDE based methods --- its results becoming statistically indistinguishable as $d=7$ is reached. }
\label{fig:dimensions}
\end{figure}

Thus we start with the peakiness problem, which has historically favored the optimization \& sampling approaches of OS* and A*. Our results with standard deviation are shown in \autoref{fig:peaky}. While ERS has a lower acceptance rate for the easiest case of $a=1$, ERS is almost uninhibited by increased peakiness as $a \to 20$, which quickly turns into a large and dramatic advantage of 56.7 percentage points. This makes ERS the most robust to the challenge by a large margin, and we argue superior to PRS and NNARS that, while better for $a = 1$, quickly drop in efficacy down to an acceptance rate of 2\% and 0.2\% respectively. 

We wish to point out that ERS works on \autoref{eq:peakiness} directly, where NNARS and PRS must have applied an unknown transformation o convert \autoref{eq:peakiness} to one with finite support. We suspect that, given the same transformation, ERS is likely to have comparable or better acceptance rate in the $a = 1$ case.

We next consider the dimension scaling problem, which historically favors KDE based approaches.  The results are shown in \autoref{fig:dimensions}. Here we see a somewhat inverted behavior, where ERS dominates for $d \leq 4$, but becomes statistically equivalent to PRS and NNARS as the dimension increases. Thus we conclude that ERS has equal or better performance in all cases for this problem.

The dimension scaling and peakiness results combined are particularly significant in that ERS performs overall best for both tasks, where previously performance favored only one between two different styles of adaptive rejection samplers. That we perform well across both tasks directly speaks to the original design goal: an easy to use sampler that can be an initial solution applied to problems and obtain effective results. In particular, we imagine that running ERS for even larger $n$ can be an effective way of determining how challenging a problem may be to devise a better rejection sampler for, since it performs well with no tuning. 

\begin{wraptable}[10]{r}{0.6\columnwidth}
\vspace{-15pt}
\centering
\caption{Clutter problem \autoref{eq:clutter} acceptance rate results for 1 and 2D data, including standard deviation  ($\sigma$) of results. } \label{tbl:clutter}
\begin{tabular}{@{}lcccc@{}}
\toprule
         & ERS  & PRS  & A*   & SRS  \\ \midrule
1D       & \textbf{95.0} & 79.5 & 89.4 & 17.6 \\
$\sigma$ & 0.7  & 0.2  & 0.8  & 0.1  \\
2D       & \textbf{92.4} & 51.0 & 56.1 & $<$0.0 \\
$\sigma$ & 1.0  & 0.4  & 0.5  & $<$0.0 \\ \bottomrule
\end{tabular}
\end{wraptable}

Last we consider the clutter problem, with results presented in \autoref{tbl:clutter}. Here we see that ERS has a 5.6\% point advantage over A* in the $d=1$ case, and a larger 36.33\% advantage in the $d=2$ case. Though ERS's advantage comes at a mild increase in the variance of the results, the large gap in acceptance rate more than makes up for the variance.

\subsection{Runtime Considerations}

Our implementation uses JAX, making it easy to add acceleration as well as use current Python tools. This is relevant because ERS is the only vectorizable algorithm under consideration. Multiple samples and computations are collected concurrently, allowing the potential for acceleration. All models were benchmarked in a Google Colab instance, which provided a Tesla V100 GPU and a TPU. We consider only the A* and OS* for runtime comparisons at baseline since the author's code is available and implements the Clutter problem. This way we are directly comparing against the original implementation details.  

\begin{wraptable}[16]{r}{0.35\columnwidth}
\vspace{-10pt}
\centering
\caption{Runtime of ERS compared to the faster A* and OS* algorithms on the Clutter problem. All times reported in seconds. } \label{tbl:runtime}
\begin{tabular}{@{}cr@{}}
\toprule
Method.                      & Clutter                       \\ \midrule
ERS: CPU                     & 1148.5                                    \\
ERS: GPU                     & 634.2                                      \\
ERS: TPU                     & 1398.6                                    \\
A*                           & 2631.5                                          \\
OS*                          & 3348.4                                           \\ \bottomrule
\end{tabular}
\end{wraptable}

The results can be found in \autoref{tbl:runtime}, where we see that ERS is the fastest in terms of runtime. The acceptance rate difference does not explain the difference in runtimes. Scaling the A* runtime by the difference in acceptance rates would give a speed of 2,476.4 seconds to produce $10^5$ samples, which is still $2.16\times$ slower than ERS for $d=1$. Adding a GPU increases the speed advantage to $3.9\times$. We note that the TPU appears to be slightly slower than CPU, and is likely due to the design of TPUs to work on larger batches of data at one time for larger neural networks. 

We note that the GPU has very low utilization in our testing, in part because $10^5$ $d=1$ samples is only 6.4 MB of data to be processed, which is not enough to fully leverage the compute capacity of these devices. As such we would anticipate even larger speedups for problems that required larger sample sizes, and efficacious GPU/TPU utilization for rejection sampling is an avenue for further research. 

\subsection{Ablating Hard-Constants}

While our approach has a number of ``magic numbers'', each is set with a simple intuition e.g., only run the GMM when 50\% more items have been sampled because there must be enough new data to get a different result. We ablate these by modifying each constant with 4 values in the range of 50\% smaller to 100\% larger than specified. We then run ESR for all combinations of these constants on the peakyness task with $a=20$ as it has the most variance of any of our tests. We use the same seed for the generated samples because we want to see what the impact of these constants are, and so all runs getting the same generated sequence of samples allows us to isolate that factor. In doing so we observe a minimum acceptance rate of 75.5\% and a maximum of 75.9\%, indicating only a 0.4\% variation due to the hard constants.

\section{Conclusion} \label{sec:conclusion}

Our approach of refining a parameterized proposal distribution $g(\cdot)$ represents a new approach to defining general purpose rejection samplers. It requires fewer parameters, functions, and derivations to be specified---requiring only the target function $f(\cdot)$ to be specified, while simultaneously returning up to $7\times$ higher acceptance rates and $4\times$ lower runtime even after accounting for the difference in acceptance rates. While we do not resolve the limitations of rejection sampling to higher dimensional data, our method enables a strong and effective baseline. 

\bibliographystyle{ecai}
\bibliography{references,refsZ,exra}

\end{document}